%% file: main.tex
\documentclass[a4paper,USenglish]{llncs}
\usepackage{graphicx}
\usepackage{amsmath}
\usepackage{fullpage}
\usepackage[numbers,sort&compress]{natbib}
\usepackage[draft,mode=multiuser]{fixme}
\fxuseenvlayout{colorsig}
\fxusetargetlayout{color}
\FXRegisterAuthor{jy}{ajy}{JY}
\FXRegisterAuthor{hb}{ahb}{HB}
\FXRegisterAuthor{si}{asi}{SI}
\FXRegisterAuthor{mt}{amt}{MT}

\usepackage[algoruled,linesnumbered,algo2e,vlined]{algorithm2e}

\usepackage{microtype}%
\newcommand{\Substr}{\mathit{Substr}}
\newcommand{\Prefix}{\mathit{Prefix}}

\newcommand{\SSuffix}{\mathit{SparseSuffix}}
\newcommand{\lcp}{\mathit{lcp}}
\newcommand{\RLEDAWG}{\mathit{RLE\_DAWG}}
\newcommand{\Endpos}{\mathit{EndPos}}
\newcommand{\RLE}{\mathit{RLE}}

\newcommand{\RLEsubstr}{\mathit{RLE\_Substr}}

\newcommand{\RLEsuffix}{\mathit{RLE\_Suffix}}

\newcommand{\RLESA}{\mathit{RLE\_SA}}
\newcommand{\RLERANK}{\mathit{RLE\_RANK}}

\newcommand{\val}{\mathit{val}}
\newcommand{\size}{\mathit{size}}

\newcommand{\maxe}{\mathrm{mpe}}
\newcommand{\longest}[1]{\overleftarrow{#1}}

\newcommand{\Insert}[2]{\mathit{Insert}(#1,#2)}
\newcommand{\Delete}[2]{\mathit{Delete}(#1,#2)}
\newcommand{\MinXInRectangle}[3]{\mathit{MinXInRectangle}(#1,#2,#3)}
\newcommand{\MaxXInRectangle}[3]{\mathit{MaxXInRectangle}(#1,#2,#3)}
\newcommand{\MaxYInRange}[2]{\mathit{MaxYInRange}(#1,#2)}

\author{
  Jun'ichi Yamamoto
  \and
  Hideo Bannai
  \and
  Shunsuke Inenaga
  \and
  Masayuki Takeda
}
\institute{
  Department of Informatics, Kyushu University
  \email{\{junichi.yamamoto,bannai,inenaga,takeda\}@inf.kyushu-u.ac.jp}\\
}

\title{
  Time and Space Efficient Lempel-Ziv Factorization
  based on Run Length Encoding
}

\begin{document}
\maketitle

\begin{abstract}
  We propose a new approach for calculating the Lempel-Ziv
  factorization of a string, based on run length encoding (RLE).
  We present a conceptually simple off-line algorithm based on 
  a variant of suffix arrays,
  as well as an on-line algorithm based on a variant of
  directed acyclic word graphs (DAWGs).
  Both algorithms run in $O(N+n\log n)$ time and $O(n)$ extra space,
  where $N$ is the size of the string, $n\leq N$ is the number of RLE factors.
  The time dependency on $N$ is only in the conversion of the string to RLE,
  which can be computed very efficiently in $O(N)$ time and $O(1)$
  extra space (excluding the output).
  When the string is compressible via RLE, i.e., $n = o(N)$, 
  our algorithms are, to the best of our knowledge, the first algorithms
  which require only $o(N)$ extra space while running in $o(N\log N)$ time.

\end{abstract}
\input{introduction}

\input{preliminaries}
\input{offline}

\input{LZ77fromDAWG}
\input{conclusion}
\bibliographystyle{splncs03}
\bibliography{ref}
\newpage
\input{appendix}
\end{document}

%% file: introduction.tex
\section{Introduction}
The \emph{run-length encoding} (RLE) of a string $S$ is a natural encoding of $S$,
where each maximal run of character $a$ of length $p$
in $S$ is encoded as $a^p$,
e.g., the RLE of string $\mathtt{aaaabbbaa}$ is $\mathtt{a^4 b^3 a^2}$.
Since RLE can be regarded as a compressed representation of strings,
it is possible to reduce the processing time and working space
if RLE strings are not decompressed while being processed.
Many efficient algorithms that deal with RLE versions of
classical problems on strings have been proposed in the literature
(e.g.:
exact pattern
matching~\cite{bunke93,apostolico99:_match_run_lengt_encod_strin,amir03:_inplac},
approximate matching~\cite{makinen03:_approx_match_run_lengt_compr_strin,apostolico12:_param},
edit distance~\cite{bunke95,arbell02:_edit,liu07:_editd,chen11:_fully_compr_algor_comput_edit},
longest common subsequence~\cite{freschi04:_longes,liu08:_findin},
rank/select structures~\cite{lee09:_dynam}, palindrome detection~\cite{chen12:_effic}).
In this paper, we consider the problem of 
computing the \emph{Lempel-Ziv factorization} (\emph{LZ factorization})
of a string via RLE.

The LZ factorization (and its variants) of a
string~\cite{LZ77,LZSS,crochemore84:_linear}, discovered over 30 years ago,
captures important properties concerning repeated occurrences
of substrings in the string,
and has applications in the field of data compression,
as well as being the key component
to various efficient algorithms on strings~\cite{kolpakov99:_findin_maxim_repet_in_word,duval04:_linear}.
Therefore, there exists a large amount of work devoted to its
efficient computation.
A na\"ive algorithm that computes the longest common prefix
with each of the $O(N)$ previous positions only requires
$O(1)$ space (excluding the output),
but can take $O(N^2)$ time, where $N$ is the length
of the string.
Using string indicies such as suffix trees~\cite{Weiner}
and on-line algorithms to construct them~\cite{Ukk95},
the LZ factorization can be computed in an on-line
manner in $O(N\log|\Sigma|)$ time and $O(N)$ space,
where $|\Sigma|$ is the size of the alphabet.
Most recent
algorithms~\cite{chen08:_lempel_ziv_factor_using_less_time_space,crochemore08:_simpl_algor_comput_lempel_ziv_factor,crochemore09:_lpf_comput_revis,a.ss:_lempel_ziv_lz77,ohlebusch11:_lempel_ziv_factor_revis}
first construct the suffix array~\cite{manber93:_suffix} of the string,
consequently taking $O(N)$ extra space and
at least $O(N)$ time, and are off-line.

Since the most efficient algorithms 
run in worst-case linear time and are practical,
it may seem that not much better can be achieved. 
However, a theoretical interest is whether or not we can achieve even
faster algorithms, at least in some specific cases.
In this paper, we propose a new approach for calculating the Lempel-Ziv
factorization of a string, which is based on its \emph{RLE}.
The contributions of this paper are as follows:
We first show that the size of the LZ encoding with self-references
(i.e., allowing previous occurrences of a factor to
overlap with itself) is at most twice as large as the size of its RLE.
We then present two algorithms that compute the LZ factorizations
of strings given in RLE:
an off-line algorithm based on suffix arrays for RLE strings,
and an on-line algorithm based on directed acyclic word graphs (DAWGs)~\cite{blumer85:_small_autom_recog_subwor_text}
for RLE strings.
Given an RLE string of size $n$, both algorithms work for 
general ordered alphabets,
and run in $O(n\log n)$ time and $O(n)$ space.
Since the conversion from a string of size $N$ to its RLE can be
conducted very efficiently in $O(N)$ time and $O(1)$ extra space
(excluding the output),
the total complexity is $O(N+n\log n)$ time and $O(n)$ space.

In the worst-case, 
the string is not compressible by RLE, i.e. $n = N$.
Thus, for integer alphabets, our approach can be slightly slower
than the fastest existing algorithms which run in $O(N)$ time (off-line) or
$O(N\log|\Sigma|)$ time (on-line).
However, for general ordered alphabets,
the worst-case complexities of our algorithms match previous algorithms
since the construction of the suffix array/suffix tree can take
$O(N\log N)$ time if $|\Sigma|=O(N)$.
The significance of our approach is
that it allows for improvements in the computational complexity
of calculating the LZ factorization for a non-trivial
family of strings; strings that are compressible via RLE.
If $n = o(N)$,
our algorithms are, to the best of our knowledge, the
first algorithms which require only $o(N)$ extra space while running
in $o(N\log N)$ time.


%

%
\vspace*{3mm}
\noindent \textbf{Related Work.} 
For computing the LZ78~\cite{LZ78} factorization,
a sub-linear time and space algorithm was presented
in~\cite{jansson07:_compr_dynam_tries_applic_lz}.
In this paper, we consider a variant of the more powerful
LZ77~\cite{LZ77} factorization.
Two space efficient on-line algorithms for LZ factorization based on succinct data
structures have been proposed~\cite{okanohara08:_onlin_algor_findin_longes_previous_factor,starikovskaya12:_comput_lempel_ziv_factor_onlin}.
The first runs in $O(N\log^3 N)$ time and 
$N\log|\Sigma| + o(N\log|\Sigma|) + O(N)$ bits of
space~\cite{okanohara08:_onlin_algor_findin_longes_previous_factor},
and the other runs in $O(N\log^2 N)$ time with $O(N\log|\Sigma|)$ bits of
space~\cite{starikovskaya12:_comput_lempel_ziv_factor_onlin}.
Succinct data structures basically simulate accesses to their non-succinct
counterparts using less space at the expense of speed.
A notable distinction of our approach is that we 
aim to reduce the {\em problem size} via compression,
in order to improve both time and space efficiency.

%% file: preliminaries.tex
\section{Preliminaries}

Let $\mathcal{N}$ be the set of non-negative integers.
Let $\Sigma$ be a finite {\em alphabet}.
An element of $\Sigma^*$ is called a {\em string}.
The length of a string $S$ is denoted by $|S|$. 
The empty string $\varepsilon$ is the string of length 0,
namely, $|\varepsilon| = 0$.
For any string $S\in\Sigma$, let $\sigma_S$ denote the number of
distinct characters appearing in $S$.
Let $\Sigma^+ = \Sigma^* - \{\varepsilon\}$.
For a string $S = XYZ$, $X$, $Y$ and $Z$ are called
a \emph{prefix}, \emph{substring}, and \emph{suffix} of $S$, respectively.
The set of prefixes of $S$ is denoted by $\Prefix(S)$.
The \emph{longest common prefix} of strings $X,Y$, denoted $\lcp(X, Y)$, 
is the longest string in $\Prefix(X) \cap \Prefix(Y)$.
The $i$-th character of a string $S$ is denoted by 
$S[i]$ for $1 \leq i \leq |S|$,
and the substring of a string $S$ that begins at position $i$ and
ends at position $j$ is denoted by $S[i..j]$ for $1 \leq i \leq j \leq |S|$.
For convenience, let $S[i..j] = \varepsilon$ if $j < i$.

For any character $a \in \Sigma$ and $p \in \mathcal{N}$,
let $a^{p}$ denote the concatenation of $p$ $a$'s,
e.g., $a^1 = a$, $a^2 = aa$, and so on.
$p$ is said to be the exponent of $a^p$.
Let $a^0 = \varepsilon$.

Our model of computation is the word RAM
with the computer word size at least $\lceil \log_2 |S| \rceil$ bits,
and hence, standard instructions on
values representing lengths and positions of string $S$
can be executed in constant time.
Space complexities will be determined by the number of computer words (not bits).

\subsection{LZ Encodings}
LZ encodings are dynamic dictionary based encodings with many variants.
As in most recent work, we describe our algorithms with respect to a
well known variant 
called s-factorization~\cite{crochemore84:_linear} 
in order to simplify the presentation.

\begin{definition}[s-factorization~\cite{crochemore84:_linear}]
  \label{def:s_factorization}
  The s-factorization of a string $S$ is
  the factorization $S = f_1 \cdots f_n$ where each
  s-factor $f_i\in\Sigma^+~(i=1,\ldots,n)$
  is defined inductively as follows:
  $f_1 = S[1]$. For $i \geq 2$:
  if $S[|f_1\cdots f_{i-1}|+1] = c \in \Sigma$ does not occur in $f_1\cdots f_{i-1}$,
  then $f_i = c$. Otherwise, $f_i$ is the longest prefix of $f_i
  \cdots f_n$ that occurs at least twice in $f_1 \cdots f_i$.
\end{definition}
Note that each s-factor can be represented in constant size,
i.e., 
either as a single character or a pair of integers representing the
position of a previous occurrence of the factor and its length.
For example the s-factorization of the string
  $S = \mathtt{abaabababaaaaabbabab}$ is
  $\mathtt{a}$,
  $\mathtt{b}$,
  $\mathtt{a}$,
  $\mathtt{aba}$,
  $\mathtt{baba}$,
  $\mathtt{aaaa}$,
  $\mathtt{b}$,
  $\mathtt{babab}$. This can be represented as
  $\mathtt{a}$, $\mathtt{b}$, $(1,1)$, $(1,3)$,
  $(5,4)$, $(10,4)$, $(2,1)$, $(5,5)$.
In this paper, we will focus on describing algorithms that output only
the length of each factor of the s-factorization, but it is not
difficult to modify them to output the previous position 
as well, in the same time and space complexities.

\subsection{Run Length Encoding}
\begin{definition}
  The Run-Length (RL) factorization of a string $S$ is
  the factorization $f_1,\ldots,f_n$ of $S$ such that for every $i=1,\ldots,n$,
  factor $f_i$ is the longest prefix of $f_i \cdots f_n$ with 
  $f_i \in \{a^p \mid a \in \Sigma, p > 0\}$.
\end{definition}

Note that each factor $f_i$ can be written as $f_i=a_i^{p_i}$ for some character 
$a_i\in\Sigma$ and some integer $p_i>0$ and
for any consecutive factors $f_i = a_i^{p_i}$ and $f_{i+1} = a_{i+1}^{p_{i+1}}$,
we have that $a_i\neq a_{i+1}$.
The {\em run length encoding} (RLE) of a string $S$, denoted
$\RLE_S$,
is a sequence of pairs consisting of a character $a_i$ and an integer
$p_i$, representing the RL factorization.
The {\em size} of $\RLE_S$ is the number of RL factors in $\RLE_S$
and is denoted by $\size(\RLE_S)$,
i.e., if $\RLE_S = a_1^{p_1} \cdots a_n^{p_n}$, then $\size(\RLE_S)  = n$.
$\RLE_S$ can be computed in $O(N)$ time and
$O(1)$ extra space (excluding the $O(n)$ space for output),
where $N = |S|$,
simply by scanning $S$ from beginning to end, counting the exponent of
each RL factor.
Also, noticing that each RL factor must consist of the same alphabet,
we have $\sigma_S \leq n$.

Let $\val$ be the function that ``decompresses'' 
$\RLE_S$, i.e., $\val(\RLE_S) = S$.
For any $1 \leq i \leq j \leq n$,
let $\RLE_S[i..j] = a^{p_i}_{i} a^{p_{i+1}}_{i+1} \cdots a^{p_j}_{j}$.
For convenience, let $\RLE_S[i..j] = \varepsilon$ if $i > j$.
Let 
$\RLEsubstr(S) = \{ \RLE_S[i..j] \mid 1 \leq i, j \leq n \}$ and
$\RLEsuffix(S) = \{ \RLE_S[i..n] \mid 1 \leq i \leq n \}$.

The following simple but nice observation allows us to represent the
complexity of our algorithms in terms of $\size(\RLE_S)$. 
\begin{lemma} \label{lem:RLE_LZ_linear}
  For a given string $S$, 
  let $n_{\mathit{RL}}$ and $n_{\mathit{LZ}}$ respectively be the
  number of factors
  in its RL factorization and s-factorization.
  Then, $n_{\mathit{LZ}} \leq 2n_{\mathit{RL}}$.
\end{lemma}
\begin{proof}
  Consider an s-factor that starts at the $j$th position
  in some RL-factor $a_i^{p_i}$ where $1 < j \leq p_i$.
  Since $a_i^{p_i - j + 1}$ is both a suffix and a prefix of $a_i^{p_i}$,
  we have that the s-factor extends at least to the end of $a_i^{p_i}$.
  This implies that a single RL-factor is always covered by
  at most 2 s-factors, thus proving the lemma.
\end{proof}
Note that for LZ factorization variants without self-references,
the size of the output LZ encoding may come into play, when it is
larger than $O(\size(\RLE_S))$.

\subsection{Priority Search Trees} \label{sec:3-sided}

In our LZ factorization algorithms,
we will make use the following data structure,
which is essentially an elegant mixture of a priority heap and balanced search tree.
\begin{theorem}[McCreight~\cite{mccreight:priority_search_tree_1985}] \label{theo:PST}
For a dynamic set $D$ which contains $n$ ordered pairs of integers,
the priority search tree (PST) data structure 
supports all the following operations and queries in $O(\log n)$ time,
using $O(n)$ space:
\begin{itemize}
 \item $\Insert{x}{y}$: Insert a pair $(x, y)$ into $D$;
 \item $\Delete{x}{y}$: Delete a pair $(x, y)$ from $D$;
 \item $\MinXInRectangle{L}{R}{B}$: Given three integers $L \leq R$ and $B$,
       return the pair $(x, y) \in D$ with minimum $x$ satisfying
       $L \leq x \leq R$ and $y \geq B$;
 \item $\MaxXInRectangle{L}{R}{B}$: Given three integers $L \leq R$ and $B$,
       return the pair $(x, y) \in D$ with maximum $x$ satisfying
       $L \leq x \leq R$ and $y \geq B$;
 \item
   $\MaxYInRange{L}{R}$: Given two integers $L \leq R$,
   return the pair $(x, y) \in D$ with maximum $y$ satisfying
   $L \leq x \leq R$.
\end{itemize}
\end{theorem}

%% file: offline.tex
\section{Off-line LZ Factorization based on RLE} \label{sec:LZ77fromSA}

In this section we present our off-line algorithm 
for s-factorization. 
The term off-line here implies that 
the input string $S$ of length $N$ is first converted to 
a sequence of RL factors, $\RLE_S = a_{1}^{p_1} a_{2}^{p_2} \cdots a_{n}^{p_n}$.
In the algorithm which follows,
we introduce and utilize RLE versions of classic string data structures.

\subsection{RLE Suffix Arrays}
Let $\Sigma_{\RLE_S} = \{ \RLE_S[i] \mid i = 1, \ldots, n\}$.
For instance, if $\RLE_S = \mathtt{a^3b^5a^3b^5a^1b^5a^4}$,
then $\Sigma_{\RLE_S} = \{\mathtt{a^1}, \mathtt{a^3}, \mathtt{a^4}, \mathtt{b^5}\}$.
For any $a_{i}^{p_{i}}, a_{j}^{p_{j}} \in \Sigma_{\RLE_S}$,
let the order $\prec$ on $\Sigma_{\RLE_S}$ be defined as
$
  a_{i}^{p_{i}} \prec a_{j}^{p_{j}} 
  \iff
  a_{i} < a_{j}, \mbox{ or } a_{i} = a_{j} \mbox{ and } p_{i}
  < p_{j}.
$
The lexicographic ordering on $\RLEsuffix(S)$ is defined over
the order on $\Sigma_{\RLE_S}$,
and our RLE version of suffix arrays~\cite{manber93:_suffix}
is defined based on this order: 
\begin{definition}[RLE suffix arrays]
  For any string $S$, its \emph{run length encoded suffix array}, denoted 
  $\RLESA_S$, is an array of length $n = \size(\RLE_S)$ such that
  for any $1 \leq i \leq n$,
  $\RLESA_S[i] = j$ when $\RLE_S[j..n]$ is the lexicographically
  $i$-th element of $\RLEsuffix(S)$.
\end{definition}

Let $\SSuffix(S) = \{\val(s) \mid s \in \RLEsuffix(S)\}$,
namely, $\SSuffix(S)$ is the set of ``uncompressed'' RLE suffixes of string $S$.
Note that the lexicographic order of $\RLEsuffix(S)$
represented by $\RLESA_S$ is not necessarily
equivalent to the lexicographic order of $\SSuffix(S)$.
In the running example, $\RLESA_S = [5, 3, 1, 7, 4, 2, 6]$.
However, the lexicographical order for the elements in $\SSuffix(S)$
is actually $(7, 1, 3, 5, 6, 2, 4)$.

\begin{lemma} \label{lem:RLESA}
  Given $\RLE_S$ for any string $S\in\Sigma^*$,
  $\RLESA_S$ can be constructed in $O(n\log n)$ time, 
  where $n = \size(\RLE_S)$.
\end{lemma}
\begin{proof}
  Any two RL factors can be compared in
  $O(1)$ time, so the lemma follows from algorithms such as in~\cite{larsson99:_faster_suffix_sortin}.
\end{proof}
Let $\RLERANK_S$ be an array of length $n = \size(\RLE_S)$
such that 
$
 \RLERANK_S[j] = i \iff \RLESA_S[i] = j.
$
Clearly $\RLERANK_S$ can be computed in $O(n)$ time 
provided that $\RLESA_S$ is already computed.
To make the notations simpler, in what follows we will denote
$rs(h) = \RLESA_S[h]$ and $rr(h) = \RLERANK_S[h]$
for any $1 \leq h \leq n$.

For any RLE strings $\RLE_X$ and $\RLE_Y$ with $\val(\RLE_X) = X$ and $\val(\RLE_Y) = Y$,
let $\lcp(\RLE_X,$ $\RLE_Y) = \lcp(X,Y)$,
i.e., $\lcp(\RLE_X,\RLE_Y)$ is the longest prefix of the ``uncompressed'' strings $X$ and $Y$.
It is easy to see that $Z=\lcp(\RLE_X,\RLE_Y)$ can be computed in $O(\size(\RLE_Z))$ time by
a naive comparison from the beginning of $\RLE_X$ and $\RLE_Y$,
adding up the exponent $p_i$ of the RL factors while $a_i^{p_i} = \RLE_X[i] = \RLE_Y[i]$,
and possibly the smaller exponent of the first
mismatching RL factors,
provided that they are exponents of the same character.




%
%
The following two lemmas imply an interesting and useful property of 
our data structure; although $\RLESA_S$ does not necessarily correspond
to the lexicographical order of the uncompressed RLE suffixes,
RLE suffixes that are closer to a given entry in the $\RLESA_S$ 
have longer longest common prefixes with that entry.

\begin{lemma} \label{lem:RLELCP_1}
Let $i,j$ be any integers such that $1 \leq i < j \leq n$.
(1) For any $j^\prime > j$, 
$|\lcp(\RLE_S[rs(i)..n],$ $\RLE_S[rs(j)..n])| \geq |\lcp(\RLE_S[rs(i)..n], \RLE_S[rs(j^\prime)..n])|$.
(2) For any $i^\prime < i$, 
$|\lcp(\RLE_S[rs(i)..n],$ $\RLE_S[rs(j)..n])| \geq
|\lcp(\RLE_S[rs(i^\prime)..n], \RLE_S[rs(j)..n])|$.
\end{lemma}
\begin{proof}
We only show (1). (2) can be shown by similar arguments. Let 
\begin{eqnarray*}
 k & = & \min\{t \mid \RLE_S[rs(i)..rs(i)+t-1] \neq \RLE_S[rs(j)..rs(j)+t-1]\} \mbox{ and } \\
 k^\prime & = & \min\{t^\prime \mid \RLE_S[rs(i)..rs(i)+t^\prime-1] \neq \RLE_S[rs(j^\prime)..rs(j^\prime)+t^\prime-1]\}.
\end{eqnarray*}
Namely, the first $(k-1)$ RL factors
of $\RLE_S[rs(i)..n]$ and $\RLE_S[rs(j)..n]$ coincide and the $k$th RL factors differ.
The same goes for $k^\prime$, $\RLE_S[rs(i)..n]$, and $\RLE_S[rs(j^\prime)..n]$.
Since $j^\prime > j$, $k \geq k^\prime$.
If $k > k^\prime$, then clearly the lemma holds.
If $k = k^\prime$, 
then $\RLE_S[rs(i)+k] \prec \RLE_S[rs(j)+k] \preceq \RLE_S[rs(j^\prime)+k]$.
This implies that $|\lcp(\RLE_S[rs(i)+k], \RLE_S[rs(j)+k])| \geq 
|\lcp(\RLE_S[rs(i)+k], \RLE_S[rs(j^\prime)+k])|$.
The lemma holds since
for these pairs of suffixes, the RL factors after the $k$th do not
contribute to their $\lcp$s.
\end{proof}


\subsection{LZ factorization using $\RLESA$}
\label{subsec:lz77sa_details}
In what follows we describe our algorithm that computes 
the s-factorization using $\RLESA_S$.
Assume that we have already computed the first $(j-1)$ s-factors 
$f_1, f_2, \ldots, f_{j-1}$ of string $S$.
Let $\sum_{h=1}^{j-1}|f_{h}| = \ell-1$, i.e., 
the next s-factor $f_{j}$ begins at position $\ell$ of $S$.
Let $d = \min \{k \mid \sum_{i=1}^{k}(p_{i}) \geq \ell \} + 1$,
i.e., the $(d-1)$-th RL factor $a_{d-1}^{p_{d-1}}$ contains
the occurrence of the $\ell$-th character $S[\ell] = a_{d-1}$ of $S$.
Let $q = \sum_{i=1}^{d-1}(p_{i}) - \ell + 1$,
i.e., $\RLE_{S[\ell..N]} = a_{d-1}^{q} a_{d}^{p_d} \cdots a_{n}^{p_n}$.
Note that $1 \leq q \leq p_{d-1}$.

The task next, is to find the longest previously occurring prefix of 
$a_{d-1}^{q} a_{d}^{p_d} \cdots a_{n}^{p_n}$ which will be $f_j$.
The difficulty here, compared to the non-RLE case, is that $f_j$
will not necessarily begin at positions in $S$ corresponding to an
entry in the $\RLESA_S$.
A key idea of our algorithm is that
rather than looking directly for $a_{d-1}^{q} a_{d}^{p_d}$ $\cdots$ $a_{n}^{p_n}$,
we look for the longest previously occurring prefix 
of $\RLE_S[d..n] = a_{d}^{p_d} \cdots a_{n}^{p_n}$ 
whose occurrence is {\em immediately preceded by} $a_{d-1}^q$,
which will have corresponding entries in $\RLESA_S$.
To compute $f_j$ in this way, we use the following lemma:
\begin{lemma} \label{lemma:LZ_PST}
Assume the situation mentioned above.
Let $k = \max(\{ p_i \mid a_{i} = a_{d-1}, 1 \leq i \leq d-2 \}\cup\{0\})$.
(Case 1) If $q = p_{d-1}$ and $q > k$,
then $|f_j| = \max\{k, 1\}$.
(Case 2) Otherwise, 
\[
 |f_{j}| = q + \max\{|\lcp(\RLE_S[rs(x_1)..n], \RLE_S[d..n])|, 
                     |\lcp(\RLE_S[d..n], \RLE_S[rs(x_2)..n])|\}
\] 
where 
 \begin{eqnarray*}
  x_1 & = & \max\{u \mid 1 < rs(u) < d,~u < rr(d),~a_{rs(u)-1} = a_{d-1},~p_{rs(u)-1} \geq q\} \mbox{ and}\\
  x_2 & = & \min\{v \mid 1 < rs(v) < d,~rr(d) < v,~a_{rs(v)-1} = a_{d-1},~p_{rs(v)-1} \geq q\}. 
 \end{eqnarray*} 
\end{lemma}
\begin{proof}
Case (1) is when the new s-factor begins at the beginning of the
RL-factor $a_{d-1}^{p_{d-1}}$, and must end somewhere inside it.
Otherwise (Case (2)), $|f_j|$ is at least $q$ since
either $q < p_{d-1}$ and $a_{d-1}^q = S[\ell..\ell+q-1] = S[\ell-1..\ell+q-2]$
is a prefix of $f_j$ due to the self-referencing nature of s-factorization,
or, there exists an RL factor $a_i^{p_i}$ such that $1 \leq i \leq d-2$, $a_{d-1} = a_i$ and $q \leq p_i$.
Let $f_j = a_{d-1}^{q} X$.
Then $X$ is the longest of the longest common prefixes between $\RLE_S[d..n]$
and $\RLE_S[h..n]$ for all $1 \leq h < d$,
\emph{that are immediately preceded by $a_{d-1}^{q}$}, i.e.
$a_{h-1} = a_{d-1}$ and $p_{h-1} \geq q$.
It follows from Lemma~\ref{lem:RLELCP_1} that of all these $h$,
the one with the longest lcp with $\RLE_S[d..n]$ is
either of the entries that are closest to position $rr(d)$ in the
suffix array.
The positions $x_1$ and $x_2$ of these entries can be described
by the equation in the lemma statement, and thus the lemma holds.
\end{proof}

Once $x_1$ and $x_2$ of the above Lemma are determined,
our algorithm is similar to conventional non-RLE algorithms that use
the suffix array.
The main difficulty lies in computing these values,
since, unlike the non-RLE algorithms, 
they depend on the value $q$ which is determined only during the
s-factorization. The main result of this section follows.

\begin{theorem} \label{theo:nlogn_LZ}
Given $\RLE_S$ of any string $S\in\Sigma^*$,
we can compute the s-factorization of $S$ in $O(n \log n)$
time and $O(n)$ space where $n = \size(\RLE_S)$.
\end{theorem}
\begin{proof}
First, compute $\RLESA_S$ in $O(n\log n)$ time, and $\RLERANK_S$  in $O(n)$ time.
Next we show how to compute each s-factor $f_j$  using Lemma~\ref{lemma:LZ_PST}.

Recall that the s-factor begins somewhere in the $(d-1)$-th RL factor.
We shall maintain, for each character $a \in \Sigma$,
a PST $T_{a}$ of Theorem~\ref{theo:PST}
for the set of pairs
$U_a^{d-1} = \{ (x, y) \mid x = rr(i),~a_{i-1} = a,~y = p_{i-1},~1 < i \leq d-1\}$,
i.e., the $x$ coordinate is the position in the suffix array, 
and the $y$ coordinate is the exponent of the preceding 
character of that suffix.
Then,
we can easily check whether the condition of case (1) is satisfied
by computing the $k$ of Lemma~\ref{lemma:LZ_PST} as $k = \MaxYInRange{1}{n}$ on $T_{a_{d-1}}$, 
which can be computed in $O(\log n)$ time by Theorem~\ref{theo:PST}\footnote{Actually, this 
  is an $O(1)$ operation on the PST, since
the information for the pair with maximum $y$ in the entire range of $x$,
is contained in the root of the PST.}.
For case (2), we obtain $x_1$ and $x_2$ as:
$x_1 = \MaxXInRectangle{1}{rr(d)-1}{q}$ and
$x_2 = \MinXInRectangle{rr(d)+1}{n}{q}$
in $O(\log n)$ time, using $T_{a_{d-1}}$.
To compute the length of the lcp's,
$\lcp(\RLE_S[rs(x_1)..n],$ $\RLE_S[d..n])$
and $\lcp(\RLE_S[d..n],$ $\RLE_S[rs(x_2)..n])$,
we simply use the naive algorithm mentioned previously,
that compares each RL factor from the beginning.
Since the longer of the two longest common prefixes is adopted,
the number of comparisons is at most twice the number of RL factors 
that is spanned by the determined s-factor.
From Lemma~\ref{lem:RLE_LZ_linear}, the total number of RL factors
compared, i.e. $\sum_{i=1}^{n_{LZ}} \size(\RLE_{f_i})$, is $O(n)$.

After computing the s-factor $f_j$, we update the PSTs.
Namely, if $f_j$ spans the RL factors $a_{d}^{p_{d}} \cdots a_{d+g}^{p_{d+g}}$,
then we insert pair $(rr(i), p_{i-1})$ into $T_{a_{i-1}}$,
for all $d \leq i \leq d+g$.
These insertion operations take $O(g \log n)$ time by Theorem~\ref{theo:PST},
which takes a total of $O(n \log n)$ time for computing all $f_j$.
Hence the total time complexity is $O(n \log n)$.

We analyze the space complexity of our data structure.
Notice that a collection of sets $U_{a}^{d-1}$ for all characters $a \in \Sigma$
are pairwise disjoint,
and hence $\sum_{a \in \Sigma}|U_{a}^{d-1}| = d-1$.
By Theorem~\ref{theo:PST}, 
the overall size of the PSTs 
is $O(n)$ at any stage of $d = 1, 2, \ldots, n$.
Since $\RLESA_S$ and $\RLERANK_S$ occupy $O(n)$ space each,
we conclude that the overall space requirement of our data structure is $O(n)$.
\end{proof}

%% file: LZ77fromDAWG.tex
\section{On-line LZ Factorization based on RLE}
Next, we present an on-line algorithm
that computes s-factorization based on RLE.
The term on-line here implies that for a string $S$
of (possibly unknown) length $N$, 
the algorithm iteratively computes the output for input string
$S[1..i]$ for each $i = 1,\ldots, N$ (the output of $S[1..i]$ can be
reused to compute the output for $S[1..i+1]$).
For example, $\RLE_S$ of string $S$ of length $N$ can be computed on-line in a total
of $O(N)$ time and $O(n)$ space (including the output), where $n = \size(\RLE_S)$.
In the description of our algorithms, this definition 
will be relaxed for simplicity, and we shall work on $\RLE_S$, 
where the s-factorization of
$\val(\RLE_S[1..i])$ is iteratively computed for $i = 1,2,\ldots n$.
Note that the off-line algorithm
described in the previous section cannot be directly transformed 
to an on-line algorithm, even if we simulate the suffix array using
suffix trees, which can be constructed on-line~\cite{Ukk95}.
This is because the elements inserted into the PST depended on the
lexicographic rank of each suffix, which can change dynamically in the
on-line setting.
Nonetheless, we overcome this problem by taking a different approach,
utilizing remarkable characteristics of a string index structure called
{\em directed acyclic word graphs} ({\em DAWGs})~\cite{blumer85:_small_autom_recog_subwor_text}.

\subsection{RLE DAWGs}
The DAWG of a string $S$
is the smallest automaton that accepts all suffixes of $S$.
Below we introduce an RLE version of DAWGs:
We regard $\RLE_S$ as a string of length $n$ over 
alphabet $\Sigma_{\RLE_S} = \{ \RLE_S[i] \mid i = 1, \ldots, n\}$.
For any $u \in \RLEsubstr(S)$, let $\Endpos_{\RLE_S}(u)$ denote
the set of positions where an occurrence of $u$ ends in $\RLE_S$,
i.e.,
$\Endpos_{\RLE_S}(u) = \{j \mid u = \RLE_S[i..j], 1 \leq i \leq j \leq
n\}$ for any $u\in\Sigma^+$ and
$\Endpos_{\RLE_S}(\varepsilon) = \{ 0, \ldots, n\}$.
Define an equivalence relation for any $u, w \in \RLEsubstr(S)$ by 
$
 u \equiv_{\RLE_S} w \Longleftrightarrow \Endpos_{\RLE_S}(u) = \Endpos_{\RLE_S}(w).
$
The equivalence class of $u \in \RLEsubstr(S)$ 
w.r.t. $\equiv_{\RLE_S}$ is denoted by $[u]_{\RLE_S}$.
When clear from context, we abbreviate the above notations
as $\Endpos$, $\equiv$ and $[u]$, respectively.
Note that for any two elements in $[u]$, one is a suffix of the
other (or vice versa).
We denote by $\longest{u}$ the
longest member of $[u]$.

\begin{definition}
The \emph{run length encoded DAWG} of a string $S \in \Sigma^*$, 
denoted by $\RLEDAWG_S$, is the DAWG of $\RLE_S$ 
over alphabet $\Sigma_{\RLE_S} = \{ \RLE_S[i] \mid i = 1, \ldots, n\}$.
Namely, $\RLEDAWG_S = (V, E)$ where
$V = \{[u] \mid u \in \RLEsubstr(S)\}$ and 
$E = \{([u], a^{p}, [u a^{p}]) 
\mid u, ua^{p} \in \RLEsubstr(S),~u \not\equiv ua^{p}\}.$
\end{definition}
We also define the set $F$ of labeled reversed edges, called {\em suffix links}, 
by 
$F = \{([a^{p}u], a^{p}, [u]) \mid u, a^{p}u \in \RLEsubstr(S), u = \longest{u}\}.$
See also Fig.~\ref{fig:RLE_DAWG} that illustrates 
$\RLEDAWG_S$ for $\RLE_S = \mathtt{a^3 b^2 a^5 b^2 a^5 c^4 a^{10}}$.
Since $\Endpos(\mathtt{b^2a^5}) = \Endpos(\mathtt{a^5}) = \{3, 5\}$,
$\mathtt{b^2a^5}$ and $\mathtt{a^5}$ are represented by the same node.
On the other hand, $\Endpos(\mathtt{a^3b^2a^5}) = \{3\}$
and hence $\mathtt{a^3 b^2 a^5}$ is represented by a different node.

\begin{figure}[t]
\centerline{
 \includegraphics[width=90mm]{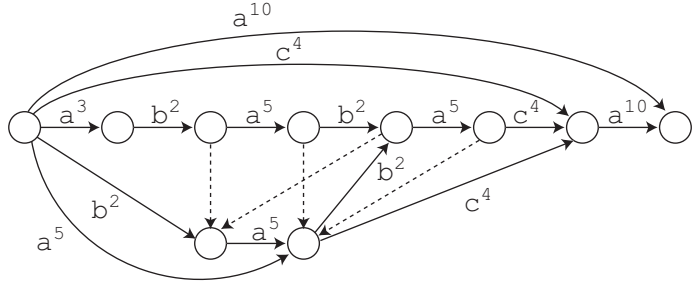}
}
\caption{Illustration for the RLE DAWG of $\mathtt{a^3 b^2 a^5 b^2 a^5 c^4 a^{10}}$.
The edges in $E$ are represented by the solid arcs, while the suffix links of some nodes are represented by dashed arcs (but their labels are omitted). For simplicity the suffix links of the other nodes are omitted in this figure.}
\label{fig:RLE_DAWG}
\end{figure}

\begin{lemma} \label{lem:RLEDAWG_size_const}
Given $\RLE_S$ of any string $S$ where $n = \size(\RLE_S)$,
$\RLEDAWG_S$ has $O(n)$ nodes and edges, and can be constructed in
$O(n \log n)$ time and $O(n)$ extra space in an on-line manner,
together with the suffix link set $F$.
\end{lemma}
\begin{proof}
  A simple adaptation of the results from~\cite{blumer85:_small_autom_recog_subwor_text}.
  (See Appendix for full proof.)
\end{proof}

\subsection{On-line LZ factorization using $\RLEDAWG$}
The high-level structure of our on-line algorithm
follows that of the off-line algorithm described in the beginning of
Section~\ref{subsec:lz77sa_details}.
In order to find the longest previously occurring prefix of
$a_{d-1}^{q} a_{d}^{p_d} \cdots a_{n}^{p_n}$, which is the next s-factor,
we construct the $\RLEDAWG_S$ on-line 
for the string up to $\RLE_S[1..d-1] =a_{1}^{p_1} a_{2}^{p_2} \cdots
a_{d-1}^{p_{d-1}}$ and use it, instead of using the $\RLESA_S$.
The difficulty is, as in the off-line case, that
only the suffixes that start at a beginning of an RL factor
is represented in the $\RLEDAWG$.
Therefore, we again look for the longest previously occurring prefix
of $\RLE_S[d..n] = a_{d}^{p_d} \cdots a_{n}^{p_n}$ that is 
immediately preceded by $a_{d-1}^q$ in $S$,
rather than looking directly for $a_{d-1}^{q} a_{d}^{p_d}$ $\cdots$ $a_{n}^{p_n}$.
We augment the $\RLEDAWG$ with some more information to make this possible.

Let $E_{[u]}$ denote the set of out-going edges of node $[u]$.
For any edge $e = ([u], b^q, [ub^q]) \in E_{[u]}$
and each character $a \in \Sigma$,
define $\maxe_e(a) = \max(\{ p \mid a^p \longest{u} b^q
\in \RLEsubstr(S)\}\cup\{0\})$.
That is, $\maxe_e(a)$ represents the maximum exponent of the
RL factor with character $a$, that precedes $\longest{u}b^q$ in $S$.
\begin{lemma} \label{lem:maxe}
Given $\RLE_S$ of any string $S\in\Sigma^*$,
$\RLEDAWG_S = (V,E)$, augmented so that
$\maxe_e(a)$ 
can be computed in $O(\log\sigma_{S})$ time
for any $e \in E$ and any character $a \in \Sigma$,
can be constructed in an on-line manner
in a total of $O(n \log n)$ time with $O(n)$ space.
\end{lemma}

\begin{proof}
When computing $\maxe_e(a)$, consider the following cases:
(Case 1) $\longest{u}b^q$ is not the longest member of $[\longest{u}b^q]$, i.e.
$\longest{u}b^q\neq\longest{ub^q}$.
For any $j \in \Endpos(\longest{u}b^q)$ let $j^\prime = j - \size(u)$.
We have that $\RLE_S[j^\prime] = a_{j^\prime}^{p_{j^\prime}}$ where
$a_{j^\prime}^{p_{j^\prime}} \longest{u}b^q \equiv \longest{u}b^q$, 
i.e., $\longest{u}b^q$ is always immediately preceded by $a_{j^\prime}^{p_{j^\prime}}$ in $\RLE_S$.
Therefore, $\maxe_e(a) = p_{j^\prime}$ if $a_{j^\prime} = a$ and $0$ otherwise.
For any node $[v]\in V$, an arbitrary $j\in\Endpos(v)$ can be easily determined in $O(1)$ time
when the node is first constructed during the on-line
construction of $\RLEDAWG_S$, and does not need to be updated.

(Case 2) $\longest{u}b^q$ is the longest member of $[\longest{u} b^q]$, i.e.
$\longest{u}b^q =\longest{ub^q}$.
For each occurrence of
$a^p\longest{u} b^q = a^p\longest{ub^q} \in \RLEsubstr(S)$, 
there must exist a suffix link $([a^p \longest{u b^q}], a^p, [\longest{u b^q}])\in F$.
Therefore $\maxe_e(a)$ is the maximum of the
exponent in the labels of all such incoming suffix links, or $0$ if there are none.
By maintaining a balanced binary search tree at every edge $e$,
we can retrieve this value for any $a\in\Sigma$ in $O(\log \sigma_S)$ time.
It also follows from the on-line construction algorithm of $\RLEDAWG_S$
that the set of labels of incoming suffix links to a node only
increases, and we can update this value in $O(\log\sigma_S)$ time for
each new suffix link.
Since $|F| = O(n)$, constructing the balanced binary search trees
take a total of $O(n \log \sigma_S)$ time, and
the total space requirement is $O(n)$.

In order to determine which case applies,
it is easy to check whether $\longest{u} b^q$ is the longest element of $[\longest{u}b^q]$
in $O(1)$ time by maintaining the length of the longest path to any given
node during the on-line construction of $\RLEDAWG_S$.
This completes the proof.
\end{proof}

\begin{lemma}\label{lem:maxmaxe}
  Given $\RLE_S$ of any string $S\in\Sigma^*$,
  $\RLEDAWG_S = (V,E)$, augmented so that
  $\max\{ p \mid \maxe_e(a) \geq q, e = ([u],b^p,[w]) \in E_{[u]} \}$
  can be computed in $O(\log n)$ time
  for any $e \in E$, character $a \in \Sigma$, and integer $q\geq 0$,
  can be constructed in an on-line manner
  in a total of $O(n \log n)$ time with $O(n)$ space.
\end{lemma}

\begin{proof}
  During the on-line construction of the augmented $\RLEDAWG_S$ of
  Lemma~\ref{lem:maxe}, we further construct and maintain a family
  of PSTs at each node of $\RLEDAWG_S$ with a total size of $O(n)$,
  containing the information to answer the query in $O(n\log n)$ time.
  (See Appendix for full proof.)
\end{proof}

The next lemma shows how the augmented $\RLEDAWG_S$ can be used 
to efficiently compute the longest prefix of a given pattern string
that appears in string $S$. 

\begin{lemma} \label{lem:longest_prefix}
  For any pattern string $P \in \Sigma^*$,
  let $\RLE_P = b_{1}^{q_1} b_{2}^{q_2} \cdots b_{m}^{q_m}$.
  Given $\RLE_P$, we can compute the length of the longest prefix $P^\prime$ of $P$ 
  that occurs in string $S$ in $O(\size(\RLE_{P^\prime}) \log n)$ time,
  using a data structure of $O(n)$ space,
  where $n = \size(\RLE_S)$.
\end{lemma}
\begin{proof}
\begin{algorithm2e}[t]
  \caption{Pattern Matching on $\RLEDAWG_S$.}
  \label{algo:rledawgmatch}
  \SetKw{KwOr}{or}
  \SetKw{KwAnd}{and}
  \SetKwData{KwSC}{shortcut}
  \SetKwData{KwTrue}{true}
  \SetKwData{KwFalse}{false}
  \SetKwData{KwCurN}{v}
  \KwIn{$\RLEDAWG_S = (V,E)$, $\RLE_P = b_1^{q_1} \cdots b_m^{q_m}$}
  $h = \max\{ q \mid ([\varepsilon], b_1^q, [w]) \in E_{[\varepsilon]}\}$\;
  \lIf{$(m = 1)$ \KwOr $(h < q_1)$}{Output $\min(h,q_1)$ and \Return \;\label{algo:firstchar}}
  $\KwSC := \KwFalse$; $\KwCurN := [\varepsilon]$\;
  \For{$i = 2,\ldots, m$}{
    \If{$e = (\KwCurN,b_i^{q_i},[w])\in E_{\KwCurN}$
      \KwAnd $(\KwSC = \KwTrue ~~\KwOr~~ \maxe_e(b_1) \geq q_1)$\label{algo:cantraverse}}{
      $\KwCurN := [w]$; \tcp{$\RLE_P[2..i] \in [w]$}
      \lIf{$\RLE_P[2..i]$ is not the longest element of $[w]$}{%
        $\KwSC := \KwTrue$ \;
      }
    }\Else{
      \tcp{$k$:maximum exponent of $b_i$ such that 
        $\val(b_{1}^{q_1} b_{2}^{q_2} \cdots b_{i}^{k}) \in\Substr(S)$.
      }
      \lIf{\KwSC = \KwTrue}{%
        $k := \max\{q \mid (\KwCurN, b_{i}^q, [w]) \in E_{\KwCurN}\}$\;
        \label{algo:k1}
      }\lElse{%
        $k := \max\{q \mid \maxe_{e}(b_1) \geq q_1, e = (\KwCurN,
        b_i^q, [w])\in E_{\KwCurN} \}$\; 
        \label{algo:k2}
      }
      Output $|\val(b_{1}^{q_1} b_{2}^{q_2} \cdots
      b_{i}^{\min\{q_i,k\}})|$ and \Return\;\label{algo:ret}
    }
  }
  Output $|P|$ and \Return; \tcp{$P$ itself occurs in $S$}
\end{algorithm2e}

The outline of the procedure is shown in Algorithm~\ref{algo:rledawgmatch}.
First, we check whether the first RL factor $b_1^{q_1}$ of $P$
is a substring of $S$ (Line~\ref{algo:firstchar}).
If so, the calculation basically proceeds by traversing
$\RLEDAWG_S$ with $b_2^{q_2}b_3^{q_3}\cdots$ until 
there is no outgoing edge with $b_i^{q_i}$
(i.e. $b_2^{q_2}\cdots b_i^{q_i} \not\in\RLEsubstr(S)$),
or, there is no occurrence of
$b_2^{q_2}\cdots b_i^{q_i}$ that is immediately preceded by $b_1^{q}$,
where $q \geq q_1$, in $S$. 
If $\textsf{shortcut} =\textsf{false}$, $b_2^{q_2}\cdots b_i^{q_i}$ is
the longest element in the node, and the latter check is conduced by
the condition $\maxe_{e}(b_1)\geq q_1$.
If $\textsf{shortcut} =\textsf{true}$, the character preceding any
occurrence of $b_2^{q_2}\cdots b_i^{q_i}$ is uniquely determined and
already checked in a previous edge traversal, so no further check is required.

By Lemmas~\ref{lem:RLEDAWG_size_const},~\ref{lem:maxe}, and~\ref{lem:maxmaxe}, 
the length of the longest prefix $P^\prime$ of $P$ that occurs in $S$ can be computed 
in $O(\size(\RLE_{P^\prime}) (\log n + \log \sigma_S)) = O(\size(\RLE_{P^\prime})\log n)$ time
using a data structure of $O(n)$ space.
\end{proof}

Below we give an example for Lemma~\ref{lem:longest_prefix}.
See Fig.~\ref{fig:RLE_DAWG} that illustrates 
$\RLEDAWG_S$ for $\RLE_S = \mathtt{a^3 b^2 a^5 b^2 a^5 c^4 a^{10}}$,
and consider searching string $S$ for pattern $P$ with $\RLE_P = \mathtt{a^{5} b^{2} a^{7}}$.
We start traversing $\RLEDAWG_S$
with the second RL factor $\mathtt{b^{2}}$ of $P$.
Since there is an out-going edge labeled $\mathtt{b}^{2}$ from the source node 
we reach node $v = [\mathtt{b^{2}}]$.
There are two suffix links that point to node $v$,
$([\mathtt{a^{3} b^{2}}], \mathtt{a^{3}}, [\mathtt{b^{2}}])$ 
and $([\mathtt{a^{5} b^{2}}], \mathtt{a^{5}}, [\mathtt{b^{2}}])$.
Hence $\maxe_{([\varepsilon],{\mathtt{b^2}},[\mathtt{b^2}])}(\mathtt{a}) = \max\{3, 5\} = 5$,
and thus the prefix $\mathtt{a^5 b^2}$ of $P$ occurs in $S$.
We examine whether a longer prefix of $P$ occurs in $S$
by considering the third RL factor $\mathtt{a^7}$.
There is no out-going edge from $v$ that is labeled $\mathtt{a^7}$,
hence the longest prefix of $P$ that occurs in $S$ is of the form 
$\mathtt{a^5 b^2 a^\ell}$ for some $\ell \geq 0$.
We consider the set $E_{v}(\mathtt{a})$ 
of out-going edges of $v$ that are labeled $\mathtt{a}^q$ for some $q$, 
and obtain $E_{v}(\mathtt{a}) = \{([\mathtt{b^2}], \mathtt{a^5}, [\mathtt{b^2 a^5}])\}$.
We have $\maxe_{([\mathtt{b^2}], \mathtt{a^5}, [\mathtt{b^2 a^5}])}(\mathtt{a}) = \max\{3, 5\} = 5$
due to the two suffix links pointing to $[\mathtt{b^2 a^5}]$.
Thus, 
the longest prefix of $P$ that occurs in $S$ is
$\mathtt{a^5 b^2 a^{\min\{7, 5\}}} = \mathtt{a^5 b^2 a^{5}}$.

\begin{theorem} \label{theo:online_lz}
  Given $\RLE_S$ for any string $S\in\Sigma^*$,
  the s-factorization of $S$ can be computed in an on-line manner
  in $O(n \log n)$ time and $O(n)$ extra space, where $\size(\RLE_S) = n$.
\end{theorem}
\begin{proof}
Assume the situation described in the first paragraph of
Section~\ref{subsec:lz77sa_details}.
In addition, assume that we have constructed $\RLEDAWG_S^{d-1}$,
the $\RLEDAWG$ (with augmentations described previously)
for $\RLE_S[1..d-1] = a_{1}^{p_1} a_{2}^{p_2} \cdots a_{d-1}^{p_{d-1}}$.
By definition,
the longest prefix $P^\prime$ of $P = S[\ell..N]$
such that $P^\prime\in\Substr(\val(a_{1}^{p_1} a_{2}^{p_2} \cdots a_{d-1}^{p_{d-1}}))$,
is a prefix of $f_{j}$.
By Lemma~\ref{lem:longest_prefix},
we can compute $|P^\prime|$ in $O(\size(\RLE_{P^\prime})\log d)$ time.
A minor technicality is when the longest previous occurrence of $f_j$
is self-referencing. This problem can be solved by simply
interleaving the traversal and update of $\RLEDAWG_S$ for each RL factor of $f_j$.
If we suppose that $f_j$ spans the RL factors $a_{d}^{p_{d}} \cdots a_{d+g}^{p_{d+g}}$,
we can traverse and update $\RLEDAWG_S^{d-1}$ to $\RLEDAWG_S^{d+g}$
in a total of $O(g \log n)$ time by Lemma~\ref{lem:maxmaxe}.
Thus, totaling for all $f_j$,
we can compute the s-factorization in $O(n \log n)$ time.
$O(n)$ space complexity follows from 
Lemmas~\ref{lem:RLEDAWG_size_const},~\ref{lem:maxe},~\ref{lem:maxmaxe}, and~\ref{lem:longest_prefix}.
For any $i>1$, the s-factorization of $\val(RLE_S[1..i-1])$ and
the s-factorization of $\val(RLE_S[1..i])$ differs only in the last
1 or 2 factors.
It is easy to see that the s-factorization of $\val(RLE_S[1..i])$ is
iteratively computed for $i = 1,\ldots, n$, 
and the computation is on-line.
\end{proof}

%% file: conclusion.tex
\section{Discussion}

We proposed off-line and on-line algorithms 
that compute a well-known variant of LZ factorization,
called s-factorization, of a given string $S$
in $O(N + n \log n)$ time using only $O(n)$ extra space,
where $N = |S|$ and $n = \size(\RLE_S)$.
After converting $S$ to $\RLE_S$ in $O(N)$ time and $O(1)$ extra space
(excluding the output), the main part of the algorithms work only on $\RLE_S$,
running in $O(n\log n)$ time and $O(n)$ space,
and therefore
can be more time and space efficient compared to
previous LZ factorization algorithms when the input strings are
compressible by RLE.
Our algorithms are theoretically significant in that
they are the only algorithms which achieve
$o(N\log N)$ time using only $o(N)$ extra space for strings with $n = o(N)$, 
thus offering a substantial improvement to the asymptotic time
complexities in calculating the s-factorization, for a non-trivial
family of strings.
Our algorithms can be easily extended to other variants of LZ factorization.
For example, let $m$ be the size of the s-factorization \emph{without} self-references
of a given string. Since Lemma~\ref{lem:RLE_LZ_linear} does not hold 
for s-factorization without self-references,
the time complexity of the algorithm is $O(N + (n+m) \log n)$. 
The working space remains $O(n)$ (excluding the output).

Since conventional string data such as natural language texts are not
usually compressible via RLE, 
the algorithms in this paper, although theoretically interesting,
may not be very practical.
However, our approach may still have
potential practical value for other types of data and objectives.
For example, a piece of music can be thought of as being
naturally expressed in RLE, where the pitch of the tone is a character,
and the duration of the tone is its run length.
Other than for the applications to string
algorithms~\cite{kolpakov99:_findin_maxim_repet_in_word,duval04:_linear},
mentioned in the Introduction,
Lempel Ziv factorization on such RLE compressible strings can be important,
due to an interesting application of compression, including LZ77 (gzip),
as a measure of distance between data, called Normalized
Compression Distance (NCD)~\cite{li04}.
NCD has been shown to be effective for various
clustering and classification tasks, including MIDI music data,
while not requiring in-depth prior knowledge of the
data~\cite{Cilibrasi05clusteringby,keogh07:_compr}.
The NCD between two strings $S$ and $T$ w.r.t. a compression
algorithm basically depends only on the compressed sizes of the strings
$S$, $T$, and their concatenation $ST$. Therefore,
efficiently computing their
s-factorizations from $\RLE_S$, $\RLE_T$, and $\RLE_{ST}$
would contribute to 
making the above clustering and classification tasks faster and more
space efficient.

Our algorithms are based on RLE variants of classical string data structures.
However, our approach does not necessarily make the use of succinct data structures impossible.
It would be interesting to explore how succinct data structures can be used in combination
with our approach to further improve the space efficiency.

%% file: appendix.tex
\section*{Appendix}

This appendix provides complete proofs 
that were omitted due to lack of space.

\vspace*{5mm}
\noindent\textbf{Lemma~\ref{lem:RLEDAWG_size_const}.}
\textit{
  Given $\RLE_S$ of any string $S$ where $n = \size(\RLE_S)$,
  $\RLEDAWG_S$ has $O(n)$ nodes and edges, and can be constructed in
  $O(n \log n)$ time and $O(n)$ extra space in an on-line manner,
  together with the suffix link set $F$.
}

\begin{proof}
The proof is a simple adaptation of the results 
from~\cite{blumer85:_small_autom_recog_subwor_text}.
The DAWG of a string of length $m$ has $O(m)$ nodes and edges.
Since $\RLEDAWG_S$ is the DAWG of $\RLE_S$ of length $n$,
$\RLEDAWG_S$ clearly has $O(n)$ nodes and edges.
If $\sigma$ is the number of distinct characters appearing in $S$,
then the DAWG of a string of length $m$ can be constructed 
in $O(m \log \sigma)$ time and $O(m)$ space, in an on-line manner,
using suffix links.
Since $|\Sigma_{\RLE_S}| \leq n$,
$\RLEDAWG_S$ with $F$ can be constructed in $O(n \log n)$ time and 
extra $O(n)$ space, on-line.
\end{proof}

\vspace*{5mm}
\noindent\textbf{Lemma~\ref{lem:maxmaxe}.} 
\textit{
  Given $\RLE_S$ of any string $S\in\Sigma^*$,
  $\RLEDAWG_S = (V,E)$, augmented so that
  $\max\{ p \mid \maxe_e(a) \geq q, e = ([u],b^p,[w]) \in E_{[u]} \}$
  can be computed in $O(\log n)$ time
  for any $e \in E$, character $a \in \Sigma$, and integer $q\geq 0$,
  can be constructed in an on-line manner
  in a total of $O(n \log n)$ time with $O(n)$ space.
}
\begin{proof}
  During the on-line construction of the augmented $\RLEDAWG$ of
  Lemma~\ref{lem:maxe}, we further construct a family of 
  PSTs at each node.
  Let $T_{u,a,b}$ denote the PST at node $u\in V$ that contains the
  set of pairs
  $U_{u,a,b} = \{ (q\prime, \maxe_e(a)) \mid  e = ([u],b^{q\prime},[w]) \in E_{[u]}, \maxe_e(a) > 0\}$,
  where $a,b\in\Sigma$. 
  By maintaining a two-level balanced binary search tree
  (a balanced binary search tree inside each node of the first
   balanced binary search tree)
   at each node,
  $T_{u,a,b}$ can be accessed for any $a,b\in\Sigma$
  in $O(\log\sigma_S)$ time. Note that empty PSTs will not be
  inserted, and hence the total space will be proportional to the number of
  elements contained in all PSTs.
  Furthermore, the number of elements in a single PST is bounded by
  $O(n)$, so $\max\{ p \mid \maxe_e(a) \geq q, e = ([u],b^p,[w]) \in E_{[u]} \}$
  can be computed as $\MaxXInRectangle{1}{|S|}{q}$ on $T_{u,a,b}$,
  in $O(\log n)$ time by Theorem~\ref{theo:PST}.

  We now bound the total number of elements in all of the PSTs.
  Recall that $\maxe_e(a) = \max(\{ p \mid a^p \longest{u} b^q
  \in \RLEsubstr(S)\}\cup\{0\})$.

When a suffix link pointing to a node $[u]$ is created,
and when an out-going edge of $[u]$ is created,
we must update the PSTs associated with $[u]$.
An edge $([u], b^p, [v])$ is called primary if 
$\size(\longest{u})+1 = \size(\longest{v})$,
and is called secondary otherwise.

First, we consider the updates of the PSTs due to the suffix links.
Suffix links are created in the following situations 
(Also refer to~\cite{blumer85:_small_autom_recog_subwor_text}
for the on-line construction algorithm of DAWGs):
\begin{enumerate}
  \item \label{item:SufLinkFromSink} %
  The suffix link of the sink node is created.
  \item \label{item:OldToNew} %
  After a node $childState$ is split, then 
  a suffix link from $childState$ to $newChildState$ is created,
  where $newChildState$ is the new node created by the node split.
\end{enumerate}

Case \ref{item:SufLinkFromSink}: Let $[v]$ be the node that is pointed by
the suffix link of the sink node.
Let $e = ([u], b^p, [v])$ be the primary edge to $[v]$,
and $a^q$ be the RL factor which corresponds to the suffix link.
If $q > \maxe_e(a)$, then delete the pair $(p, \maxe_e(a))$ from the corresponding PST 
stored in $[u]$,
and insert a new pair $(p, q)$ into the PST.

Case \ref{item:OldToNew}:
There will be no updates in the PSTs.
We will discuss the details in Case \ref{item:SplitedEdge} of edge creation.

Next, we consider the updates of the PSTs due to the edges.
Edges are created in the following situations:
\begin{enumerate}
  \item \label{item:PrimaryEdgeToSink} %
  A primary edge from the old sink $currentSink$ to the new sink $newSink$ is created.

  \item \label{item:SecondaryEdgeToSink} %
  A secondary edge to $newSink$ is created.

  \item \label{item:SplitedEdge} %
  After a node $childState$ is split, then 
  the secondary edge to $childState$ from one of its parents $parentState$ becomes
  the primary edge from $parentState$ to $newChildState$.

  \item \label{item:CopiedEdge} %
  After a node $childState$ is split, then all the outgoing edges of $childState$ are copied 
  as the outgoing edges of $newChildState$.

  \item \label{item:EdgeToSplitedNode} %
  After a node $childState$ is split, some secondary edges to $childState$ are redirected and become secondary edges to $newChildState$.

\end{enumerate}

Case \ref{item:PrimaryEdgeToSink}: 
 Although a new primary edge is created, no pairs are inserted into nor deleted from  the PST since there are no suffix links to $newSink$.

Case \ref{item:SecondaryEdgeToSink}: 
Let $[u]$ be the node from which a secondary edge to $newSink$ is created.
We then insert a pair corresponding to the secondary edge into the PST of $[u]$.

Case \ref{item:SplitedEdge}: 
In this case, the secondary edge becomes a primary edge. 
So seemingly we might need to delete the pair corresponding to the existing secondary edge 
and insert a new pair corresponding to the incoming suffix link.
However, 
both pairs are actually identical, 
and hence we need no updates in the PSTs.

Case \ref{item:CopiedEdge}: Since the copied edges are all secondary edges (see Fig. \ref{fig:DAWGsplit}), similar updates to Case \ref{item:SecondaryEdgeToSink} 
are conducted for all the copied edges.

\begin{figure}[t]
\centerline{
 \includegraphics[width=90mm]{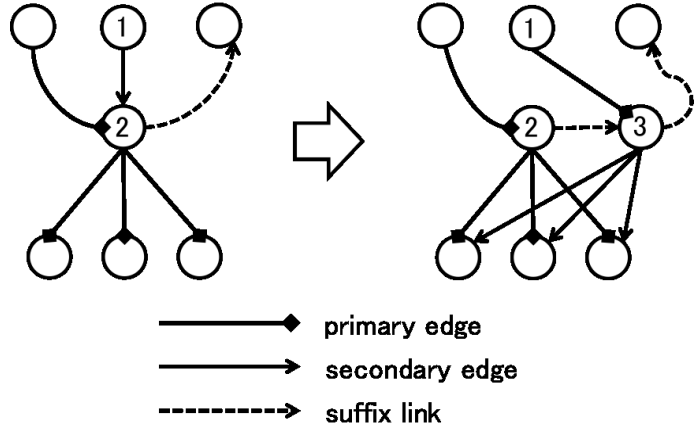}
}
\caption{Illustration for node split of the $\RLEDAWG$ construction algorithm.
Node $childState$ (depicted as $\#2$) is split 
into nodes $childState$ and $newChildState$ (depicted as $\#3$).
Note that all the outgoing edges of $newChildState$
are secondary edges.}
\label{fig:DAWGsplit}
\end{figure}

Case \ref{item:EdgeToSplitedNode}: 
Let $e$ and $e\prime$ be secondary edges 
before and after redirection, respectively.
By the property of the equivalence class,
we have that $\maxe_e(a) = \maxe_{e\prime}(a)$ for any character $a \in \Sigma$.
Hence we need no explicit updates of the PSTs.

By the above discussion, 
the number of update operations can be bounded by the number of added edges and suffix links.
Since the total number of edges and suffix links is $O(n)$, 
the total number of pairs in all of the PSTs
$\sum_{[u] \in V, a,b\in\Sigma}|U_{u,a,b}|$ is also $O(n)$.
The total time complexity for the updates is $O(n\log n)$.
\end{proof}

%% file: main.bbl
\begin{thebibliography}{10}
\providecommand{\url}[1]{\texttt{#1}}
\providecommand{\urlprefix}{URL }

\bibitem{a.ss:_lempel_ziv_lz77}
Al-Hafeedh, A., Crochemore, M., Ilie, L., Kopylov, J., Smyth, W., Tischler, G.,
  Yusufu, M.: A comparison of index-based {L}empel-{Z}iv {LZ77} factorization
  algorithms. ACM Computing Surveys  (in press)

\bibitem{amir03:_inplac}
Amir, A., Landau, G.M., Sokol, D.: Inplace run-length 2d compressed search. TCS
   290(3),  1361--1383 (2003)

\bibitem{apostolico12:_param}
Apostolico, A., Erd\"os, P.L., J\"uttner, A.: Parameterized searching with
  mismatches for run-length encoded strings. TCS  (2012)

\bibitem{apostolico99:_match_run_lengt_encod_strin}
Apostolico, A., Landau, G.M., Skiena, S.: Matching for run-length encoded
  strings. J. Complexity  15(1),  4--16 (1999)

\bibitem{arbell02:_edit}
Arbell, O., Landau, G.M., Mitchell, J.S.: Edit distance of run-length encoded
  strings. IPL  83(6),  307--314 (2002)

\bibitem{blumer85:_small_autom_recog_subwor_text}
Blumer, A., Blumer, J., Haussler, D., Ehrenfeucht, A., Chen, M.T., Seiferas,
  J.: The smallest automaton recognizing the subwords of a text. TCS  40,
  31--55 (1985)

\bibitem{bunke93}
Bunke, H., Csirik, J.: An algorithm for matching run-length coded strings.
  Computing  50,  297--314 (1993)

\bibitem{bunke95}
Bunke, H., Csirik, J.: An improved algorithm for computing the edit distance of
  run length coded strings. IPL  54,  93--96 (1995)

\bibitem{chen08:_lempel_ziv_factor_using_less_time_space}
Chen, G., Puglisi, S., Smyth, W.: {L}empel-{Z}iv factorization using less time
  \& space. Mathematics in Computer Science  1(4),  605--623 (2008)

\bibitem{chen11:_fully_compr_algor_comput_edit}
Chen, K.Y., Chao, K.M.: A fully compressed algorithm for computing the edit
  distance of run-length encoded strings. Algorithmica  (2011)

\bibitem{chen12:_effic}
Chen, K.Y., Hsu, P.H., Chao, K.M.: Efficient retrieval of approximate
  palindromes in a run-length encoded string. TCS  432,  28--37 (2012)

\bibitem{Cilibrasi05clusteringby}
Cilibrasi, R., Vit\'{a}nyi, P.M.B.: Clustering by compression. IEEE
  Transactions on Information Theory  51,  1523--1545 (2005)

\bibitem{crochemore84:_linear}
Crochemore, M.: Linear searching for a square in a word. Bulletin of the
  European Association of Theoretical Computer Science  24,  66--72 (1984)

\bibitem{crochemore09:_lpf_comput_revis}
Crochemore, M., Ilie, L., Iliopoulos, C.S., Kubica, M., Rytter, W., Wale\'n,
  T.: {LPF} computation revisited. In: Proc. IWOCA 2009. pp. 158--169 (2009)

\bibitem{crochemore08:_simpl_algor_comput_lempel_ziv_factor}
Crochemore, M., Ilie, L., Smyth, W.F.: A simple algorithm for computing the
  {L}empel {Z}iv factorization. In: Proc. DCC 2008. pp. 482--488 (2008)

\bibitem{duval04:_linear}
Duval, J.P., Kolpakov, R., Kucherov, G., Lecroq, T., Lefebvre, A.: Linear-time
  computation of local periods. TCS  326(1-3),  229--240 (2004)

\bibitem{freschi04:_longes}
Freschi, V., Bogliolo, A.: Longest common subsequence between
  run-length-encoded strings: a new algorithm with improved parallelism. IPL
  90(4),  167--173 (2004)

\bibitem{jansson07:_compr_dynam_tries_applic_lz}
Jansson, J., Sadakane, K., Sung, W.K.: Compressed dynamic tries with
  applications to {LZ}-compression in sublinear time and space. In: Proc.
  FSTTCS 2007. pp. 424--435 (2007)

\bibitem{keogh07:_compr}
Keogh, E., Lonardi, S., Ratanamahatana, C.A., Wei, L., Lee, S.H., Handley, J.:
  Compression-based data mining of sequential data. Data Mining and Knowledge
  Discovery  14(1),  99--129 (2007)

\bibitem{kolpakov99:_findin_maxim_repet_in_word}
Kolpakov, R., Kucherov, G.: Finding maximal repetitions in a word in linear
  time. In: Proc. FOCS 1999. pp. 596--604 (1999)

\bibitem{larsson99:_faster_suffix_sortin}
Larsson, N.J., Sadakane, K.: Faster suffix sorting. Tech. Rep. LU-CS-TR:99-214
  [LUNFD6/(NFCS-3140)/1--20/(1999)], Dept of Computer Science, Lund University,
  Sweden (1999)

\bibitem{lee09:_dynam}
Lee, S., Park, K.: Dynamic rank/select structures with applications to
  run-length encoded texts. TCS  410(43),  4402--4413 (2009)

\bibitem{li04}
Li, M., Chen, X., Li, X., Ma, B., Vit\'anyi, P.M.B.: The similarity metric.
  IEEE Transactions on Information Theory  50(12),  3250--3264 (2004)

\bibitem{liu07:_editd}
Liu, J., Huang, G., Wang, Y., Lee, R.: Edit distance for a run-length-encoded
  string and an uncompressed string. IPL  105(1),  12--16 (2007)

\bibitem{liu08:_findin}
Liu, J., Wang, Y., Lee, R.: Finding a longest common subsequence between a
  run-length-encoded string and an uncompressed string. J. Complexity  24(2),
  173--184 (2008)

\bibitem{makinen03:_approx_match_run_lengt_compr_strin}
M\"{a}kinen, V., Ukkonen, E., Navarro, G.: Approximate matching of run-length
  compressed strings. Algorithmica  35(4),  347--369 (2003)

\bibitem{manber93:_suffix}
Manber, U., Myers, G.: Suffix arrays: A new method for on-line string searches.
  SIAM J.~Computing  22(5),  935--948 (1993)

\bibitem{mccreight:priority_search_tree_1985}
McCreight, E.M.: Priority search trees. SIAM J. Comput.  14(2),  257--276
  (1985)

\bibitem{ohlebusch11:_lempel_ziv_factor_revis}
Ohlebusch, E., Gog, S.: {L}empel-{Z}iv factorization revisited. In: Proc. CPM
  2011. pp. 15--26 (2011)

\bibitem{okanohara08:_onlin_algor_findin_longes_previous_factor}
Okanohara, D., Sadakane, K.: An online algorithm for finding the longest
  previous factors. In: Proc. ESA 2008. pp. 696--707 (2008)

\bibitem{starikovskaya12:_comput_lempel_ziv_factor_onlin}
Starikovskaya, T.: Computing {L}empel-{Z}iv factorization online. In: Proc.
  MFCS 2012. pp. 789--799 (2012)

\bibitem{LZSS}
Storer, J., Szymanski, T.: Data compression via textual substitution. Journal
  of the ACM  29(4),  928--951 (1982)

\bibitem{Ukk95}
Ukkonen, E.: On-line construction of suffix trees. Algorithmica  14(3),
  249--260 (1995)

\bibitem{Weiner}
Weiner, P.: Linear pattern-matching algorithms. In: Proc. of 14th IEEE Ann.
  Symp. on Switching and Automata Theory. pp. 1--11 (1973)

\bibitem{LZ77}
Ziv, J., Lempel, A.: A universal algorithm for sequential data compression.
  IEEE Transactions on Information Theory  IT-23(3),  337--343 (1977)

\bibitem{LZ78}
Ziv, J., Lempel, A.: Compression of individual sequences via variable-length
  coding. IEEE Transactions on Information Theory  24(5),  530--536 (1978)

\end{thebibliography}
